\documentclass[11pt]{article}

\usepackage{amsmath}
\usepackage{amsfonts}
\usepackage{amsthm}
\usepackage{amssymb}
\usepackage{fullpage}
\usepackage{dsfont}
\usepackage{hyperref}
\usepackage{graphicx}

\newtheorem{theorem}{Theorem}
\newtheorem{lemma}[theorem]{Lemma}

\theoremstyle{definition}
\newtheorem{defn}{Definition}

\newtheorem*{thm1}{Theorem 1}

\newcommand{\abs}[1]{\left\lvert #1 \right\rvert}

\newcommand{\complex}{{\mathds C}}
\newcommand{\reals}{{\mathds R}}
\newcommand{\ints}{{\mathds Z}}
\newcommand{\rats}{{\mathds Q}}

\def\poly{\operatorname{poly}}

\newcommand{\tensor}{\otimes}

\newcommand{\adjoint}{\dagger}

\newcommand{\ket}[1]{|#1\rangle}
\newcommand{\bra}[1]{\langle #1|}
\newcommand{\ketbra}[2]{\ket{#1}\!\bra{#2}}        % outer product
\newcommand{\set}[1]{{\left\{#1\right\}}}          % braces for set notation
\newcommand{\norm}[1]{\left\|#1\right\|}           % norm
\newcommand{\enorm}[1]{\norm{#1}_{\mathrm{2}}}      % Euclidean norm for vectors
\newcommand{\onorm}[1]{\norm{#1}_{\mathrm{1}}}     % one norm for vectors

\newcommand{\trace}{{\rm Tr}}
\newcommand{\transpose}{{\mathrm T}}
\newcommand{\ve}[1]{\mathbf{#1}}
\newcommand{\enc}[1]{\left<#1\right>}

\newcommand{\ec}{\enorm{\hat{\ve{c}}}}

\newcommand{\WMEM}{\operatorname{WMEM}_\beta (\SMN)}
\newcommand{\WMEMEB}{\operatorname{WMEM}_{\beta^\prime} (\EB)}

\newcommand{\WMEMK}{\operatorname{WMEM}_\beta (K)}
\newcommand{\WMEMKp}{\operatorname{WMEM}_\beta(K)}
\newcommand{\WSEPp}{\operatorname{WSEP}_\nu(K)}
\newcommand{\WOPT}{\operatorname{WOPT}_\epsilon (\SMN)}

\newcommand{\WOPTK}{\operatorname{WOPT}_\epsilon (K)}
\newcommand{\WOPTKp}{\operatorname{WOPT}_\epsilon(K)}
\newcommand{\WVAL}{\operatorname{WVAL}_\alpha (\SMN)}

\newcommand{\Q}{\operatorname{QUSEP}}
\newcommand{\C}{\operatorname{CLIQUE}}
\newcommand{\R}{\operatorname{RSDF}}
\newcommand{\EBP}{\operatorname{EBP}}
\newcommand{\SMN}{\mathcal{S}_{M,N}}

\newcommand{\HMN}{H_{M,N}}

\newcommand{\func}[1]{#1}

\newcommand{\dens}{\mathcal{D}(\complex^M\otimes\complex^N)}
\newcommand{\denstwo}{\mathcal{D}(\complex^2\otimes\complex^{M}\otimes\complex^N)}

\newcommand{\spa}[1]{\mathcal{#1}}
\newcommand{\EB} {\mathcal{G}_{M,N}}
\newcommand{\EBtwo} {\mathcal{G}_{2M,N}}
\newcommand{\D}{\mathcal{F}(\complex^M\otimes\complex^N)}

\begin{document}
\thispagestyle{empty}
%%%%%%%%

\title{Strong NP-Hardness of the Quantum Separability Problem}

\author{Sevag Gharibian \footnote{School of Computer Science
and Institute for Quantum Computing, University of Waterloo, Waterloo, Canada}}
\maketitle
\thispagestyle{empty}

\begin{abstract}
Given the density matrix $\rho$ of a bipartite quantum state, the quantum separability problem asks whether $\rho$ is entangled or separable. In 2003, Gurvits showed that this problem is NP-hard if $\rho$ is located within an inverse exponential (with respect to dimension) distance from the border of the set of separable quantum states. In this paper, we extend this NP-hardness to an inverse polynomial distance from the separable set. The result follows from a simple combination of works by Gurvits, Ioannou, and Liu. We apply our result to show (1) an immediate lower bound on the maximum distance between a bound entangled state and the separable set (assuming $\rm{P}\neq\rm{ NP}$), and (2) NP-hardness for the problem of determining whether a completely positive trace-preserving linear map is entanglement-breaking.
\end{abstract}

%==============================================
\section{Introduction}
%==============================================

Once dubbed ``the characteristic trait of quantum mechanics''~\cite{S35}, the phenomenon of quantum entanglement has been (theoretically) harnessed in a multitude of quantum computational tasks and areas, including quantum teleportation~\cite{BBCJPW93}, superdense coding~\cite{BW92}, quantum parallelism (e.g.\ as in Shor's factoring algorithm~\cite{S94}), quantum communication complexity~\cite{CB97,G97}, and quantum cryptography~\cite{E91}. In response, there have been a number of proposed entanglement detection criteria and measures, such as the positive partial transpose (PPT) criterion~\cite{P96,HHH96}, von Neumann entropy of reduced states~\cite{DHR02}, entanglement of formation~\cite{BDSW96}, relative entropy of entanglement~\cite{VPRK97}, entanglement cost~\cite{BDSW96}, distillable entanglement~\cite{BDSW96}, negativity~\cite{VW02}, and the first need for positive but not completely positive maps in physics~\cite{HHH96}, to name but a few (see~\cite{Bru02,HHH07,BZ06} for surveys).
Yet, the problem of determining whether an arbitrary quantum state is entangled or not (where in the latter case, the state is called \emph{separable}), dubbed the Quantum Separability problem ($\Q$), was proven NP-hard by Gurvits in 2003~\cite{Gur03}.

Let us discuss how one formulates $\Q$ in a slightly more formal manner. Let $\dens$ denote the set of bipartite density operators acting on $\complex^M\otimes\complex^N$, where $M$ and $N$ denote the dimensions of the respective subsystems. A quantum state $\rho\in\dens$ is called \emph{separable} if and only if it can be written
\begin{equation}
    \rho = \sum_k p_k \ketbra{a_k}{a_k}\otimes\ketbra{b_k}{b_k},
\end{equation}
for unit vectors $\ket{a_k}\in\complex^M$ and $\ket{b_k}\in\complex^N$, and real vector $\ve{p}$ such that $p_k\geq 0$ for all $k$ and $\sum_k p_k=1$. The latter constraint implies that the set of separable density operators, which we denote by $\SMN$, is a convex set, being the convex hull of all pure product states $\ket{a}\otimes\ket{b}\in\complex^M\otimes\complex^N$. Intuitively, $\Q$ is thus the problem of determining whether a given state $\rho\in\dens$ is in $\SMN$. However, a problem arises when we encode our state $\rho$ using a computer --- due to the constraint of finite precision, we cannot in general encode the density matrix of $\rho$ \emph{exactly}. In particular, if $\rho$ sits on the border of $\SMN$, it may be that the slightly perturbed density matrix we actually encode now sits slightly outside of $\SMN$. This makes the problem ill-defined.

To circumvent this problem, one solution is to allow a margin of error in the vicinity of the border of $\SMN$. This formulation is known as \emph{Weak} Membership. Roughly, the Weak Membership problem over a convex set $K\subseteq\reals^m$ (denoted $\WMEMK$, and defined formally in Section~\ref{scn:def}) asks to decide whether a given point $\ve{y}\in\reals^m$ is in $K$, with the proviso that an algorithm is allowed to err on points lying within some fixed Euclidean distance $\beta > 0 $ from the border of $K$. Thus, to make $\Q$ well-defined, we consider the formulation $\WMEM$. Observe that $\SMN\not\subseteq\reals^m$ as required for $\WMEMK$ --- we deal with this explicitly in Section~\ref{scn:def} by more correctly letting $\SMN$ denote the set of real Bloch vectors~\cite{Kim03} corresponding to the elements of $\dens$.

In 2003, Gurvits showed~\cite{Gur03} that $\WMEM$ is NP-hard via a polynomial time Turing reduction from the NP-complete problem $\rm{PARTITION}$. Intuitively, a Turing reduction describes how to solve a problem $A$ (e.g.\ $\rm{PARTITION}$) by running an algorithm for a second problem $B$ (e.g.\ $\WMEM$) possibly multiple times. $\rm{PARTITION}$ is defined as the problem of deciding whether a finite set of integers can be partitioned into two sets of equal sum. However, $\rm{PARTITION}$ is known to be NP-hard only if the magnitudes of the input integers are exponentially large with respect to input length --- otherwise, the problem can be solved efficiently using a dynamic programming approach~\cite{GJ79}. It follows, as observed by Aaronson and later documented by Ioannou~\cite{Iou07}, that the reduction of Ref.~\cite{Gur03} shows NP-hardness for $\WMEM$ only when $\beta\leq1/\operatorname{exp}(M,N)$, i.e.\ when the input state is allowed to be exponentially close to the border of $\SMN$.

In an attempt to strengthen this result, Gurvits then devised (as explained in~\cite{Iou07}; see also~\cite{GO09}) the following reduction from the NP-complete problem $\C$ (defined in Section~\ref{scn:def}):

\begin{equation}\label{eqn:failedchain}
    \C\leq_m\R\leq_m\WVAL\leq_T \WMEM.
\end{equation}
Here, $\R$ is the problem Robust Semidefinite Feasibility (defined in Section~\ref{scn:def}), $\WVAL$ is the problem Weak Validity~\cite{Gro88} (which intuitively asks one to decide whether a given hyperplane is a separating hyperplane for a given convex set modulo some error $\alpha>0$, and whose precise definition is not needed here), $\leq_T$ denotes a Turing reduction, and $\leq_m$ denotes a \emph{many-one} reduction. A \emph{many-one} reduction is a special case of a Turing reduction in which the algorithm for problem B is invoked only \emph{once}, the output of which is immediately returned as the output for problem A. Unfortunately, the link $\WVAL\leq_T \WMEM$ above is based on the Yudin-Nemirovskii theorem~\cite{YN76}, which uses the shallow-cut ellipsoid method, and also results in exponential scaling for $\beta$~\cite{Iou07}. Thus, this reduction again shows NP-hardness of $\WMEM$ only for $\beta\leq 1/\exp(M,N)$.

The main result we show in this article is as follows.

\begin{theorem}\label{thm:strong}
    $\WMEM$ is NP-hard for $\beta\leq1/\poly(M,N)$, i.e.\ is strongly NP-hard.
\end{theorem}

A problem is called \emph{strongly NP-hard} if it is NP-hard even if the magnitudes of its numerical parameters are polynomially bounded in the length of its input~\cite{GJ79}. To show Theorem~\ref{thm:strong}, our observation is that we can replace the component $\WVAL\leq_T \WMEM$ in Eq.~(\ref{eqn:failedchain}) with a recent non-ellipsoidal Turing reduction of Liu~\cite{Liu07} (see also Ref.~\cite{Ber04}, as discussed in Ref.~\cite{Liu07}\footnote{Both Ref.~\cite{Liu07} and Ref.~\cite{Ber04} only pertain to the third reduction in Eq.~\eqref{eqn:reductionchain} --- i.e.\ they discuss the reduction $\WOPTK \leq_T \WMEMK$ for arbitrary $K$.}~) from the problem Weak Optimization (denoted $\WOPTK$, and defined in Section~\ref{scn:def}) to $\WMEMK$. We thus use the new reduction chain:
\begin{equation}
    \C \leq_m \R \leq_m \WOPT \leq_T \WMEM. \label{eqn:reductionchain}
\end{equation}
To make this work, our technical contribution is the reduction $\R \leq_m \WOPT$, which uses ideas similar to those in the reduction $\R\leq_m\WVAL$~\cite{Iou07}.

This article is organized as follows. In Section~\ref{scn:def}, we introduce all necessary definitions and notation. Section~\ref{scn:reduction} presents the proof of Theorem~\ref{thm:strong}. In Section~\ref{scn:application}, we discuss two applications of Theorem~\ref{thm:strong}. We first apply the positive partial transpose (PPT) entanglement detection criterion~\cite{P96,HHH96} to obtain immediate lower bounds on the maximum Euclidean distance between a bound entangled~\cite{HHH98} state and $\SMN$ (assuming $\rm{P}\neq\rm{ NP}$). We next use the Jamio{\l}kowski isomorphism~\cite{J72} to show NP-hardness of the problem of determining whether a completely positive trace-preserving linear map (i.e.\ a quantum \emph{channel}) is entanglement-breaking~\cite{HSR03}. We conclude in Section~\ref{scn:conclusion}.

%=================================================
\section{Definitions and Notation}\label{scn:def}
%=================================================
In this section, we formally define the following problems needed to show Theorem~\ref{thm:strong}: $\C$, Robust Semidefinite Feasibility ($\R$), Weak Optimization ($\WOPTK$), and Weak Membership ($\WMEMK$). All norms are taken as the Euclidean norm  $\enorm{~}$ (where $\enorm{A}$ corresponds to the Frobenius norm if $A$ is a matrix). The letter $\rats$ indicates the rational numbers. The notation $:=$ is used to indicate a definition. We denote (column) vector $v$ by $\ve{v}$, its conjugate transpose as $\ve{v}^\adjoint$, and its $i$th entry as $v_i$. We use the notation $\enc{\alpha}$ to signify the number of bits necessary to encode an entity $\alpha$. Specifically, if $\alpha=a/b$ is rational, we define $\enc{\alpha}=\enc{a}+\enc{b}$, and for matrix $A$, we let $\enc{A}=\sum_{ij}\enc{{A}_{ij}}$ (similarly for vectors).

First, the NP-complete problem $\C$ is stated as follows.

\begin{defn}[$\C$]\label{def:CLIQUE} Given a simple graph $G$ on $n$ vertices, and $c\leq n$, for $n,c\in \ints^+$, decide, with respect to the complexity measure $\enc{G}+\enc{c}$:
    \begin{tabbing}
        \quad\quad\= If the number of vertices in the largest complete subgraph of $G$ is at least $c$, output YES.\+\\
        Otherwise, output NO.
    \end{tabbing}
\end{defn}

\noindent Here, we take $\enc{G}=\enc{A_G}$, where $A_G$ is the $n\times n$ adjacency matrix for $G$, such that $A_G[i,j]=1$ if vertices $i$ and $j$ are connected by an edge, and $A_G[i,j]=0$ otherwise. Next, we have the problem Robust Semidefinite Feasibility.

\begin{defn}[Robust Semidefinite Feasibility ($\R$)]\label{def:RSDF} Given $k$ rational, symmetric $l\times l$ matrices $B_1,\ldots, B_k$, and $\zeta,\eta\in\rats$, with $\zeta,\eta\geq 0$, define $g(B_1,\ldots,B_k):=\max_{\ve{x}\in\reals^l,\enorm{\ve{x}}=1}\sum_{i=1}^{k}(\ve{x}^\transpose B_i \ve{x})^2$. Then, decide, with respect to the complexity measure $lk+\sum_{i=1}^{k}\enc{B_i}+\enc{\zeta}+\enc{\eta}$:
\begin{tabbing}
        \quad\quad\= If $g(B_1,\ldots,B_k)\geq \zeta+\eta$, output YES.\+\\
        If $g(B_1,\ldots,B_k)\leq \zeta-\eta$, output NO.
        \end{tabbing}
\end{defn}

We have assumed $\zeta \geq 0$ without loss of generality above, since $g(B_1,\ldots,B_k)\geq 0$. This will be necessary later in Lemma~\ref{l:link2}, when we need to take $\sqrt{g(B_1,\ldots,B_k)}$. We have also defined $\R$ as a \emph{promise} problem, meaning we are promised the input will fall into one of two disjoint cases which may be separated by a non-zero gap, and we must distinguish between the two cases. One could equivalently lift the promise and allow input falling in the ``gap'' region (i.e.\ $\zeta-\eta<g(B_1,\ldots,B_k)< \zeta+\eta$) --- for any such input, we would consider any output to be correct (i.e.\ YES or NO).

\begin{figure}\centering
  \includegraphics[height=40mm]{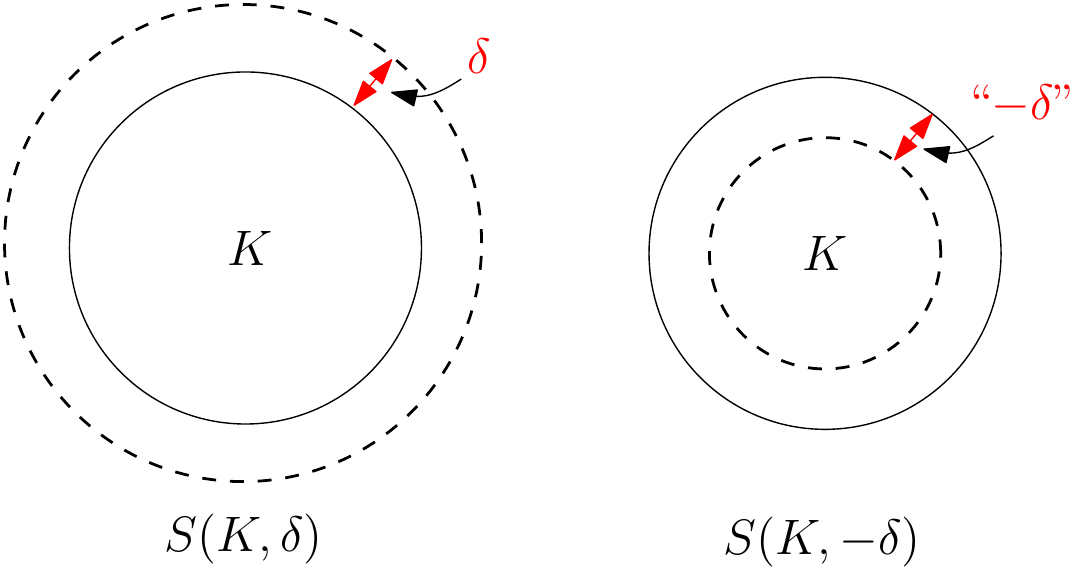}
  \caption{An intuitive picture of the sets $S(K,\delta)$ and $S(K,-\delta)$, respectively. We use the term ``$-\delta$'' in the latter illustration to stress the fact that $K$ is contracted.}\label{fig:skd}
\end{figure}

Moving on, in order to discuss the Weak Optimization and Weak Membership problems, we first require the following definitions. Let $K\subseteq \reals^m$ be a convex and compact set, and define the sets $S(K,\delta) := \{\ve{x}\in \reals^m \mid \exists~\ve{y} \in K \text{ s.t. } \enorm{\ve{x}-\ve{y}}\leq \delta\}$, and $S(K,-\delta):= \{\ve{x}\in K \mid S(\ve{x},\delta) \subseteq K \}$. Roughly, $S(K,\delta)$ can be thought of as extending the border of $K$ by $\delta$ in Euclidean distance, and $S(K,-\delta)$ can be thought of as taking the core of $K$, which is $\delta$ away from the border of $K$. The two sets are depicted in Figure~\ref{fig:skd}. As per Ref.~\cite{Gro88}, we shall require that $K$ be a \emph{well-bounded p-centered} set, meaning that $K\subseteq S(\ve{0},R)$ for $\ve{0}\in\reals^m$ the origin and some rational $R>0$, and $S(\ve{p},r)\subseteq K$ for a known point $\ve{p}\in K$ and some rational $r>0$. This ensures $K$ is full-dimensional and bounded.

Finally, we set the encoding size of $K$ as $\enc{K}=m+\enc{r}+\enc{R}+\enc{\ve{p}}$~\cite{Gur03}. The presence of parameter $m$ in this expression allows us to perform vector addition and scalar multiplication on elements of $K$ in polynomial time. We can now formally define Weak Optimization and Weak Membership over the set $K$ as follows.

\begin{defn}[Weak Optimization ($\WOPTK$)\cite{Liu07}]\label{def:WOPT}Let $K\subseteq \reals^m$ be a convex, compact, and well-bounded p-centered set. Then, given $\ve{c}\in \rats^m$ with $\enorm{\ve{c}}=1$, and $\gamma,\epsilon\in \rats$, such that $\epsilon>0$, decide, with respect to the complexity measure $\enc{K}+\enc{\ve{c}}+ \enc{\gamma}+\enc{\epsilon}$:
        \begin{tabbing}
        \quad\quad\=
        If there exists $\ve{y}\in S(K,-\epsilon)$ with $\ve{c}^\transpose \ve{y}\geq \gamma + \epsilon$, then output YES.\+\\
        If for all $\ve{x}\in S(K,\epsilon)$, $\ve{c}^\transpose \ve{x}\leq \gamma - \epsilon$, then output NO.
        \end{tabbing}
\end{defn}

\begin{defn}[Weak Membership ($\WMEMK$)]\label{def:WMEM}Let $K\subseteq \reals^m$ be a convex, compact, and well-bounded p-centered set. Then, given $\ve{y}\in\rats ^m$, and error parameter $\beta\in\rats$, such that $\beta>0$, decide, with respect to the complexity measure $\enc{K}+ \enc{\ve{y}}+\enc{\beta}$:
        \begin{tabbing}
        \quad\quad\=
        If $\ve{y}\in S(K,-\beta)$, then output YES.\+\\
        If $\ve{y}\not \in S(K,\beta)$, then output NO.
        \end{tabbing}
\end{defn}

Both of these problems are also stated as promise problems. Roughly speaking, the former asks whether there exists a point in the ``core'' of $K$ that achieves a threshold value slightly higher than $\gamma$ for a linear function defined by $\ve{c}$. The latter asks to distinguish whether a given point $\ve{y}$ is in the ``core'' of $K$ or ``far away'' from $K$. We remark that unlike here, in Ref.~\cite{Liu07} the inputs to $\WOPTK$ and $\WMEMK$ are real (as opposed to rational), and specified using $\poly(m)$ bits of precision. The latter is because such precision suffices if one demands $\epsilon$ and $\beta$ to be at least inverse polynomial in the input size~\cite{Liu07}, i.e.\ if one allows at least ``moderate'' error (which we will also demand here). It is easy to see that we can exactly represent any such $\poly(m)$-bit real numbers as rational numbers in poly-time using $\poly(m)$ bits as well, and hence the case of Ref.~\cite{Liu07} is a special case of our definitions here.

With the definitions above in place, our goal in Section~\ref{scn:reduction} is to show the reduction of Eq.~(\ref{eqn:reductionchain}), i.e.\ that an instance of $\C$ can be solved in polynomial time given an algorithm for $\WMEM$. Before proceeding, it remains for us to verify that the set of separable quantum states in $\dens$ satisfies our requirements on $K\subseteq\reals^m$ in the definitions of $\WOPTK$ and $\WMEMK$. To do so, we first represent our quantum states as real vectors in $\reals^m$, for $m$ set as follows. Let $\HMN$ denote the set of Hermitian operators mapping $\complex^M\tensor\complex^N\mapsto\complex^M\tensor\complex^N$. Then, via the isomorphism between $\HMN$ and $\reals^{M^2N^2}$, we can write any $\rho\in\dens$ as (e.g.\ \cite{Kim03}):
    \begin{equation}\label{eqn:densityToBloch}
        \rho=\frac{I}{MN}+\frac{1}{2}\sum_{i=1}^{M^2N^2-1}r_i\sigma_i,
    \end{equation}
where we have chosen as a basis for $\HMN$ the identity and the traceless Hermitian generators of $SU(MN)$, the latter denoted by $\sigma_i$ and such that $\trace(\sigma_i\sigma_j)=2\delta_{ij}$ (for $\delta_{ij}$ the Kronecker delta)~\cite{Kim03}. The vector $\ve{r}\in\reals^{M^2N^2-1}$ is called the \emph{Bloch vector} of $\rho$, whose terms are given by $r_i=\trace(\rho\sigma_i)$. Henceforth when referring to $\SMN$, we shall mean the convex set of Bloch vectors in $\reals^{M^2N^2-1}$ corresponding to separable states in $\dens$. Having represented $\SMN$ in terms of real vectors, we require it to be p-centered and well-bounded. This follows since $\SMN$ is contained in an origin-centered ball of radius $R=\sqrt{2(MN-1)/MN}$~\cite{BZ06}, and contains an origin-centered ball of radius $r=\sqrt{2/MN(MN-1)}$~\cite{Gur02}, where the origin $\ve{0}\in\reals^{M^2N^2-1}$ corresponds to the maximally mixed state, and $R$ and $r$ are with respect to the Euclidean norm. We remark that the extra factor of $\sqrt{2}$ appearing in the expressions for $r$ and $R$ is due to switching from the density matrix to the Bloch vector representation (see Lemma~\ref{l:appendix2} in Appendix~\ref{app:A}) --- this scaling will not affect our analysis. We further require that $\SMN$ be compact, which holds since the set of pure product states is closed and bounded, and the convex hull of a convex compact set is also compact~\cite{V89}.

Finally, for concreteness, we remark that from the values of $m$, $r$, $R$, and $\ve{p}$ above for $\SMN$, it follows that for the definitions of $\WOPT$ and $\WMEM$ (Def.~\ref{def:WOPT} and~\ref{def:WMEM}, respectively), $\enc{\SMN}=m+\enc{R}+\enc{r}+\enc{\ve{p}}\leq\poly(MN)$.

%===================================================
\section{The Reduction}\label{scn:reduction}
%===================================================
We now show Theorem~\ref{thm:strong} by demonstrating the polynomial-time Turing reduction of Eq.~(\ref{eqn:reductionchain}). Since $\C$ is NP-complete, this will imply NP-hardness of $\WMEM$. In addition, we must show that in Eq.~(\ref{eqn:reductionchain}), one can choose $\epsilon,\beta\geq 1/\poly(M,N)$ --- the former is required for the reduction to run in polynomial time (due to the run-time of Theorem~\ref{thm:link4}), and the latter is required to prove \emph{strong} NP-hardness (since the case of $\beta\leq 1/\rm{exp}(M,N)$ is already known to be NP-hard~\cite{Gur03}). We proceed by considering each link of Eq.~(\ref{eqn:reductionchain}) in order. As will be discussed, the first and last links are provided by Ref.~\cite{Iou07} and Ref.~\cite{Liu07}, respectively. Our job is to show the second link.

The first link in Eq.~(\ref{eqn:reductionchain}) is given to us by the following theorem~\cite{Iou07}. Unless otherwise stated, by a poly-time reduction, we mean with respect to the encoding size of the problem instance, as defined in Section~\ref{scn:def}. We use the notation $\Pi=(\text{\emph{input parameters}})$ to denote an instance $\Pi$ of a given problem, with $\Pi$ specified by the given input parameters.

\begin{theorem}[Gurvits and Ioannou~\cite{Iou07}]\label{thm:link1}
     There exists a poly-time many-one reduction which maps instance $\Pi_1=(G,n,c)$ of $\C$ to instance $\Pi_2=(k,l,B_1,\ldots,B_k,\zeta,\eta)$ of $\R$, such that $k=n(n-1)/2$, $l=n$, $B_i\in\rats^{n\times n}$ and $\enorm{B_i}\in \Theta(1)$ for all $1\leq i \leq k$, $\zeta=\Theta(1)$, $\eta\in\Omega(n^{-2})$.
\end{theorem}

We refer the reader to Ref.~\cite{Iou07} for the details of the proof of this theorem, but highlight that it relies heavily on the following theorem of Motzkin and Straus~\cite{MS65}, which ties the maximum clique number of a graph to optimization of a square-free quadratic form over the unit simplex:
    \begin{theorem}[Motzkin and Straus~\cite{MS65}]\label{thm:cliqueToWMQS}
        Denote by $(i,j)\in G$ an edge in graph $G$ between vertices $i$ and $j$, and let $\omega$ be the order of the maximal complete graph contained in $G$. Let $\Delta_n$ denote the simplex $\Delta_n:=\{\ve{x}\in\reals^n\mid x_i\geq 0\text{, }\onorm{\ve{x}}=1\}$. Then
        \begin{equation}\label{eqn:cliqueToWMQS}
            \max_{\ve{x}\in\Delta_n}\sum_{(i,j)\in G}x_ix_j=\frac{1}{2}\left(1-\frac{1}{\omega}\right).
        \end{equation}
    \end{theorem}
\noindent For later reference, we remark that the matrices $B_i\in\reals^{n\times n}$ from Theorem~\ref{thm:link1} will have the following structure --- to each $B_i$, we uniquely assign an index $(s,t)$ from the adjacency matrix $A_G$ of $G$, $1\leq s < t \leq n$, such that $B_i$ has all entries zero, except for entries $(s,t)$ and $(t,s)$, which are set to the $(s,t)$th entry of $A_G$. We hence require $k=n(n-1)/2$ matrices $B_i$, as claimed by Theorem~\ref{thm:link1}.

We now demonstrate the second link in Eq.~(\ref{eqn:reductionchain}).

\begin{lemma}\label{l:link2}
    There exists a poly-time many-one reduction which maps instance $\Pi_1=(k,l,B_1,\ldots,B_k,\zeta,\eta)$ of $\R$ to instance $\Pi_2=(\ve{c},\gamma)$ of $\WOPT$, where we define for convenience $\Delta:=\sqrt{2\sum_{i=1}^{k}\enorm{B_i}^2}$, and such that:
     \begin{itemize}
        \item $M=k+1$
        \item $N=\frac{l(l-1)}{2}+1$
        \item $\ve{c}=\ve{\hat{c}}/\enorm{\hat{\ve{c}}}$ for some $\ve{\hat{c}}\in\rats^m$ with $\enorm{\hat{\ve{c}}}\in O(m^{1/2}\Delta)$ and $m=M^2N^2-1$
        \item $\gamma=\frac{1}{2\enorm{\hat{\ve{c}}}}(\sqrt{\zeta+\eta}+\sqrt{\zeta-\eta})$
        \item $\epsilon\leq \frac{\sqrt{\zeta+\eta}-\sqrt{\zeta-\eta}}{4\enorm{\hat{\ve{c}}}(MN-1)+1}$
     \end{itemize}
\end{lemma}

\begin{proof}
    The heart of the mapping from $\Pi_1$ to $\Pi_2$ is given in Refs.~\cite{Gur03,Iou07}, and involves rephrasing the function $g(B_1,\ldots,B_k)$ from $\R$ in terms of convex optimization over the set of separable density matrices acting on $\complex^M\otimes\complex^N$. We briefly summarize this here for later reference. Let $M=k+1$, $N=\frac{l(l-1)}{2}+1$, and consider the matrix $C\in\reals^{MN\times MN}$, such that
    \begin{equation}\label{eqn:Cmatrix}
        C:=\left(
             \begin{array}{cccc}
               0 & A_1 & \ldots & A_{M-1} \\
               A_1 & 0 & \ldots & 0 \\
               \vdots & \vdots & \ddots & \vdots \\
               A_{M-1} & 0 & \ldots & 0 \\
             \end{array}
           \right),
    \end{equation}
    where each $A_i\in\reals^{N\times N}$ is symmetric and all zeroes except for its upper-left $l\times l$-dimensional submatrix, which we set to $B_i$. It is easy to see that $\enorm{C}=\Delta$, as defined in the statement of our claim. One can then write (Proposition 6.5 in Ref.~\cite{Gur03}):
    \begin{equation}\label{eqn:gurvref}
        \sqrt{g(B_1,\ldots,B_k)}=
        \max_{\ve{x}\in\reals^N,\enorm{\ve{x}}=1}\sqrt{\sum_{i=1}^{M-1}(\ve{x}^\transpose A_i \ve{x})^2} =
        \max_{\rho_{\rm sep}\in\dens}\trace(C\rho_{\rm sep}),
    \end{equation}

    \noindent where $\rho_{\rm sep}$ denotes a separable density matrix. Thus, $\Pi_1$ is reduced to optimizing the linear objective function $\trace(C\rho)$ over all separable density matrices $\rho_{\rm sep}\in\dens$. This concludes the referenced work of~\cite{Gur03,Iou07}.

    What remains is to explicitly rephrase the problem in terms of Bloch vectors and apply simple convex geometric arguments to complete the reduction, as well as characterize scaling of the error parameter $\epsilon$. To do so, first use Eq.~(\ref{eqn:densityToBloch}) and the fact that $\trace(C)=0$ to write:
    \begin{equation}\label{eqn:densityToBloch2}
        \trace(C\rho)=\trace\left(C\left(\frac{I}{MN}+\frac{1}{2}\sum_{i=1}^{M^2N^2-1}r_i\sigma_i\right)\right)=\frac{1}{2}\sum_{i=1}^{M^2N^2-1}r_i\cdot \trace(C\sigma_i)=\ve{\hat{c}}^\transpose\ve{r},
    \end{equation}
    for $\sigma_i$ the generators of $SU(MN)$, ${\hat{c}_i}:=\frac{1}{2}\trace(C\sigma_i)$, and $\ve{r}$ the Bloch vector of $\rho$. Set $m=M^2N^2-1$, $\ve{c}=\ve{\hat{c}}/ \enorm{\ve{\hat{c}}}$. In terms of Bloch vectors, our objective function $\trace(C\rho)$ in Eq.~(\ref{eqn:gurvref}) can hence be rephrased as $\func{f}(\ve{r}):=\ve{c}^\transpose\ve{r}$, with $\func{f}_{\max}:=\max_{\ve{r}\in\SMN}\func{f}(\ve{r})$. We remark that unless $C$ is the zero matrix (i.e.\ the input graph to $\C$ has no edges), we have $\enorm{\hat{\ve{c}}}>0$. Also, since $\trace(\sigma_i\sigma_j)=2\delta_{ij}$, it follows from Eq.~(\ref{eqn:densityToBloch2}) and the Cauchy-Schwarz inequality that $\enorm{\hat{\ve{c}}}\in O(m^{1/2}\Delta)$.

    To complete the reduction, we now must show the following (for $\gamma$ and $\epsilon$ to be chosen as needed): If $\func{f}_{\max}\geq \ec^{-1}\sqrt{\zeta+\eta}$, then there exists an $\ve{r}\in S(\SMN,-\epsilon)$ such that $\func{f}(\ve{r}) \geq \gamma + \epsilon$ (i.e.\ a YES instance of $\R$ implies a YES instance of $\WOPT$). If $\func{f}_{\max}\leq \ec^{-1}\sqrt{\zeta-\eta}$, then for all $\ve{r}\in S(\SMN,\epsilon)$, $\func{f}(\ve{r}) \leq \gamma - \epsilon$ (i.e.\ a NO instance of $\R$ implies a NO instance of $\WOPT$). The $\ec^{-1}$ term in these expressions follows from our definition of $\ve{c}$, and the square root in $\sqrt{\zeta+\eta}$ follows from the square root in Eq.~(\ref{eqn:gurvref}). We proceed case by case. Set $\gamma=\frac{1}{2\ec}(\sqrt{\zeta+\eta}+\sqrt{\zeta-\eta})$, and let us choose $\epsilon$ as needed.

    \begin{itemize}
    \item Case 1: $\func{f}_{\max}\geq \frac{1}{\ec}\sqrt{\zeta+\eta}$.\vspace{2mm} \\
        Let $\ve{r}^\ast\in\SMN$ be such that $\func{f}(\ve{r}^\ast)=\func{f}_{\max}$. To find an $\ve{r}\in S(\SMN,-\epsilon)$ such that $\func{f}(\ve{r})\geq \gamma + \epsilon$, we first use the fact that for any well-bounded origin-centered convex set $K\subseteq\reals^m$, it holds that for all $\ve{x}\in K$, there exists a $\ve{y}\in S(K,-\epsilon)$ such that $\enorm{\ve{x}-\ve{y}}\leq2\epsilon R/r$~\cite{Gro88} (where $R$ and $r$ are the radii of the ball containing $K$ and the origin-centered ball contained within $K$, respectively). Plugging in the definitions of $r$ and $R$ for $\SMN$ from Section~\ref{scn:def}, it follows that there exists an $\ve{r}\in S(\SMN,-\epsilon)$ such that $\enorm{\ve{r}-\ve{r}^\ast}\leq2(MN-1)\epsilon$. Since $\func{f}$ is linear, we can write:
            \begin{equation}
                \abs{\func{f}(\ve{r})-\func{f}(\ve{r}^\ast)}=\abs{\ve{c}^\transpose(\ve{r}-\ve{r}^\ast)}\leq
                \enorm{\ve{c}}\enorm{\ve{r}-\ve{r}^\ast}\leq2(MN-1)\epsilon,\label{eqn:fdistdiff}
            \end{equation}
            where the first inequality follows from the Cauchy-Schwarz inequality. Thus, in order to have $\func{f}(\ve{r})\geq \gamma + \epsilon$ as desired, it suffices to have
            \begin{equation}
                \func{f}(\ve{r})\geq \func{f}_{\max}-2(MN-1)\epsilon\geq\gamma + \epsilon,
            \end{equation}
            into which substitution of our values for $\gamma$ and $\func{f}_{\max}$ gives that setting
            \begin{equation}\label{eqn:eps1}
                \epsilon\leq\frac{\sqrt{\zeta+\eta}-\sqrt{\zeta-\eta}}{4\ec(MN-1)+1}
            \end{equation}
            suffices to conclude we have a YES instance of $\WOPT$.
    \item Case 2: $\func{f}_{\max}\leq \frac{1}{\ec}\sqrt{\zeta-\eta}$.\vspace{2mm} \\
    We must show that for an appropriate choice of $\epsilon$, we have $\func{f}(\ve{r})\leq \gamma - \epsilon$ for all $\ve{r}\in S(\SMN,\epsilon)$. Set $\ve{r}\in S(\SMN,\epsilon)$. Then by the definition of $S(\SMN,\epsilon)$, there exists some $\ve{r}^\prime\in\SMN$ such that $\enorm{\ve{r}-\ve{r}^\prime}\leq\epsilon$. By Eq.~(\ref{eqn:fdistdiff}), it follows that \begin{equation}\label{eqn:case2}
        \abs{\func{f}(\ve{r})-\func{f}(\ve{r}^\prime)}\leq\epsilon.
    \end{equation}
    In the worst case, one has $\func{f}(\ve{r}^\prime)=\func{f}_{\max}$. Hence, by Eq.~(\ref{eqn:case2}), we have that $\func{f}(\ve{r})\leq \func{f}_{\max}+\epsilon$ for any $\ve{r}\in S(\SMN,\epsilon)$. To achieve $\func{f}(\ve{r})\leq \gamma - \epsilon$ then, set $\func{f}(\ve{r})\leq \func{f}_{\max}+\epsilon \leq \gamma-\epsilon$, into which substitution of our values for $\gamma$ and $\func{f}_{\max}$ yields that choosing
    \begin{equation}
        \epsilon \leq \frac{\sqrt{\zeta+\eta}-\sqrt{\zeta-\eta}}{2\ec+2}
    \end{equation}
    suffices to conclude we have a NO instance of $\WOPT$.
    \end{itemize}
\end{proof}
Observe that combining Theorem~\ref{thm:link1} and Lemma~\ref{l:link2} gives  $M=N=\frac{n(n-1)}{2}+1$. Following an argument of Ioannou (Section 2.2.5 of Ref.~\cite{Iou07}), one can likewise show Lemma~\ref{l:link2} for $M\geq N$ by padding the matrix $C$ from its proof with extra $N\times N$-dimensional zero matrices. Thus, the hardness result we show for $\WMEM$ by building on this link will be valid for $M\geq N$.

The last link of Eq.~(\ref{eqn:reductionchain}) is given by applying the non-ellipsoidal Turing reduction of Liu~\cite{Liu07}, which holds for an arbitrary p-centered well-bounded compact convex set $K\subseteq\reals^m$.

\begin{theorem}[Proposition 2.8 of Ref.~\cite{Liu07}]\label{thm:link4}
    Let $K\subseteq \reals^m$ be a convex, compact, and well-bounded p-centered set with associated radii $(R,r)$. Given an instance $\Pi=(K,\ve{c},\gamma,\epsilon)$ of $\WOPTKp$, with $0<\epsilon<1$, there exists an algorithm which runs in time $\poly(\enc{K},R,\lceil1/\epsilon\rceil)$, and solves $\Pi$ using an oracle for $\WMEMKp$ with $\beta=r^3\epsilon^3/[2^{13} 3^3 m^5R^4(R+r)]$.
\end{theorem}

We briefly review some of the ideas behind the proof~\cite{Liu07} of Theorem~\ref{thm:link4}, which builds on results of Gr\"{o}tschel \emph{et al}~\cite{Gro88}. $\WOPTKp$ is first reduced to a problem called \emph{Weak Separation} ($\WSEPp$) over $K$, which roughly asks one to determine whether a point $\ve{p}\in\reals^m$ is approximately in $K$, and if not, to return a hyperplane approximately separating $\ve{p}$ from $K$. This first reduction is achieved in two steps: First, use the given oracle for $\WSEPp$ to construct a weak separation oracle for the set $K^\prime := \set{\ve{y}\in K \mid \ve{c}^T \ve{y} \geq \gamma}$. Second, use the latter oracle in an iterative approach where, in each iteration, we check if a candidate point $\ve{p}$ is in $K^\prime$, and if not, we update $\ve{p}$ using the separating hyperplane returned by the oracle. $\WSEPp$ is then reduced to $\WMEMKp$ via three further reductions, the latter two of which are given by Lemmas 4.3.3 and 4.3.4 of Ref.~\cite{Gro88}. Lemma 4.3.3, in particular, demonstrates that a single call to a Weak Membership oracle with two-sided error (i.e.\ the kind we use in this article) can be used to solve a variant of Weak Membership with only one-sided error, i.e.\ where in the NO case of Def.~\ref{def:WMEM}, one affirms that $\ve{y}\not\in K$, as opposed to $\ve{y}\not\in S(K,\beta)$.

We now show that composing the three reductions above yields a polynomial time reduction from $\C$ to $\WMEM$, such that $\beta\geq 1/\poly(M,N)$. First, observe the dependence on $\lceil1/\epsilon\rceil$ in the runtime stated in Theorem~\ref{thm:link4}. We thus must be able to choose $\epsilon\geq1/\poly(m)\geq1/\poly(M,N)$ in order for our reduction chain of Eq.~(\ref{eqn:reductionchain}) to run in polynomial time, where $m=M^2N^2-1$ for $\SMN$. Let us show that we can solve our instance of $\C$ with such a choice of $\epsilon$. Specifically, by Lemma~\ref{l:link2}, we can set
\begin{equation}\label{eqn:epsvalue}
    \epsilon= \frac{\sqrt{\zeta+\eta}-\sqrt{\zeta-\eta}}{4\ec(MN-1)+1}.
\end{equation}

Piecing together Lemma~\ref{thm:link1} and Lemma~\ref{l:link2}, we first immediately have $\zeta\in \Theta(1)$ and $\eta\in\Omega(1/N)$. It follows that for some positive constants $c_1$, $c_2$, $N_1$, and $N_2$, we have:
\begin{eqnarray}
    \sqrt{\zeta+\eta}-\sqrt{\zeta-\eta}&\geq&\sqrt{c_1+\frac{c_2}{N}}-\sqrt{c_1-\frac{c_2}{N}}\quad\quad\quad\forall ~N\geq\max(N_1,N_2)\\
        &=&\frac{(c_1+\frac{c_2}{N})-(c_1-\frac{c_2}{N})}{\sqrt{c_1+\frac{c_2}{N}}+\sqrt{c_1-\frac{c_2}{N}}}\\
        &\geq&\frac{2c_2}{N(\sqrt{c_1+c_2}+\sqrt{c_1})}.
\end{eqnarray}
Hence, $\sqrt{\zeta+\eta}-\sqrt{\zeta-\eta}\in\Omega(1/N)$. With a little thought, we also have that $\ec\in O(\sqrt{N})$ (see Appendix~\ref{app:A}, Lemma~\ref{l:appendix2}). Plugging these bounds into Eq.~(\ref{eqn:epsvalue}) yields that we can solve an instance of $\C$ with $\epsilon\in\Omega(M^{-1}N^{-5/2})$. Combining Theorem~\ref{thm:link1}, Lemma~\ref{l:link2}, and Theorem~\ref{thm:link4}, we thus have a polynomial time Turing reduction from $\C$ to $\WMEM$.

To show that this implies \emph{strong} NP-hardness of $\WMEM$, it remains to ensure that $\beta\geq1/\poly(M,N)$. Substituting our values for $r$, $R$, and $m$ for $\SMN$ from Section~\ref{scn:def} into the expression for $\beta$ in Theorem~\ref{thm:link4}, we have $\beta=\poly(M^{-1},N^{-1},\epsilon)$. By our choice of $\epsilon$ above, we thus have $\beta\geq1/\poly(M,N)$, as required. We hence conclude:

\begin{theorem}\label{thm:finalthm}
    Given instance $\Pi=(G,n,c)$ of $\C$, there exists an algorithm which solves $\Pi$ in time $\poly(n)$ using an oracle for $\WMEM$ with parameters $M=N=\frac{n(n-1)}{2}+1$ and some $\beta\in\Omega(n^{-73})$. More generally, for $N=\frac{n(n-1)}{2}+1$ and any choice of $M\geq N$, the result holds for some $\beta\in\Omega(M^{-16}N^{-20.5})$.
\end{theorem}

\begin{thm1}[\emph{revised}]
    $\WMEM$ is NP-hard for $\beta\leq \poly(M^{-16}N^{-20.5})$ and $M\geq N$, or equivalently, is strongly NP-hard.
\end{thm1}

We stress the phrase ``\emph{some} $\beta\in\Omega(n^{-73})$'' in the statement of Theorem~\ref{thm:finalthm} above --- specifically, we cannot have $\beta\in \Omega(1)$ in our reduction, due, for example, to the expression for $\epsilon$ in Lemma~\ref{l:link2}. The question of NP-hardness of $\WMEM$ for $\beta\in \Omega(1)$ thus remains open. We remark that the major contributor to the large negative exponent on $n$ in the estimate for $\beta$ is the reduction of Theorem~\ref{thm:link4}.

%==========================================================================================
\section{Applications}~\label{scn:application}
%==========================================================================================
We propose two applications of Theorem~\ref{thm:strong}, both of which to our knowledge were not previously known. In this section, for simplicity of exposition we revert back to letting $\SMN$ denote the set of density matrices corresponding to separable quantum states in $\dens$. By Lemma~\ref{l:appendix2} in Appendix~A, distances in the respective spaces are asymptotically equivalent. Our first application, which benefits directly from the improved NP-hardness results of Theorem~\ref{thm:strong}, is an immediate lower bound on the maximum distance a bound entangled state can have from $\SMN$, assuming $\rm{P}\neq\rm{ NP}$. To see this, recall that bound entangled states are mixed entangled quantum states from which no pure (state) entanglement can be distilled~\cite{HHH98}, and are the only entangled states whose entanglement is capable of escaping detection by the Peres-Horodecki Positive Partial Transpose (PPT)~\cite{P96,HHH96} criterion\footnote{The converse question of whether \emph{all} bound entangled states escape detection by the PPT test is, however, a major open question~\cite{Bru02}.}. Theorem~\ref{thm:strong} implies that, unless $P=NP$, any test of membership for $\SMN$ must be unable to efficiently resolve $\SMN$ within distance $\beta\in\Omega(M^{-16}N^{-20.5})$ of its border in the general case. It follows that unless $P=NP$, there must exist bound entangled state(s) $\rho_{be}\in \dens$ such that for any separable state $\rho_{sep}\in\SMN$, $\enorm{\rho_{be}-\rho_{sep}}\in\Omega(M^{-16}N^{-20.5})$ --- if not, one could determine the separability of any quantum state within this region efficiently using the PPT test, contradicting Theorem~\ref{thm:strong}. Further improvements to our hardness estimate for $\beta$ would directly benefit this application.

Our second application is to show that the problem of determining whether a completely positive (CP) trace-preserving (TP) linear map (i.e.\ a quantum channel) is entanglement-breaking (EB) is NP-hard. We remark that here we do not use the improved NP-hardness results of Theorem~\ref{thm:strong} (i.e.\ NP-hardness of $\WMEM$ for inverse exponential $\beta$ would suffice for our proof below). At the end of this section, we briefly discuss how the improved hardness bounds of Theorem~\ref{thm:strong} may help extend the result here to \emph{strong} NP-hardness of determining whether a channel is EB. Returning to our discussion, quantum channels are of interest, as they correspond to physically realizable processes. EB channels in particular are well-studied~\cite{Ruskai03,King02,Shor02,King03,HolevoSW05,Holevo08}, and are intuitively defined as the set of channels whose action on half of any bipartite state \emph{always} results in a separable state, i.e.\ they \emph{break} entanglement across the bipartition. Our approach is to show that the ability to determine whether a quantum channel is EB suffices to solve $\Q$, implying the former must be NP-hard.

We begin by stating the key ingredients required to study our problem, and then sketch a proof of the desired result. To simplify our discussion, we do not use Weak Membership formulations here --- the use of such formulations to make our argument rigorous is described in the ensuing discussion. Let $L(\spa{X})$ denote the set of linear operators mapping a complex Euclidean space $\spa{X}$ to itself. Then, EB channels are defined as follows.
\begin{defn}[Horodecki, Shor, and Ruskai~\cite{HSR03}]\label{def:EBmapsreal}
    A quantum channel $\Phi:L(\complex^{M^\prime})\mapsto L(\complex^M)$ is \emph{entanglement-breaking} (EB) if $(\Phi\otimes I)(\rho)\in\SMN$ for all $\rho\in D(\complex^{M^\prime}\otimes\complex^{N})$.
\end{defn}

Next, informally, let $\EBP$ denote the problem of determining whether a channel $\Phi$ given as input is EB. How do we encode $\Phi$ as input? For this, we choose\footnote{Although we represent our channels here in Jamio{\l}kowski form, it is straightforward to move to another representation, such as the Kraus representation. For example, the canonical way to determine a Kraus representation for $\Phi$ given its Jamio{\l}kowski representation $J(\Phi)$ is to reshuffle the eigenvectors of $J(\Phi)$ into matrices (i.e.\ Kraus operators) $K_i$. We refer the reader to~\cite{BZ06} for further details.} the Jamio{\l}kowski representation~\cite{J72} of maps, since it known that this representation provides a direct link to the separability of density operators~\cite{HSR03}, as will be stated shortly in Theorem~\ref{def:EBmaps}. The Jamio{\l}kowski representation is defined as follows. For any $\Phi:L(\complex^{N})\mapsto L(\complex^M)$, there exists a unique operator, denoted $J(\Phi)\in L(\complex^M\otimes\complex^N)$, obtained by applying the Jamio{\l}kowski isomorphism to $\Phi$, the action of which is given by
\begin{equation}\label{eqn:jam_eqn}
    J(\Phi) := \left[\Phi\tensor I\right]\left(\ketbra{\phi^+}{\phi^+}\right),
\end{equation}
where $\ket{\phi^+}$ is the maximally entangled state $\ket{\phi^+}=\frac{1}{\sqrt{N}}\sum_{k=0}^{N-1}\ket{k}\tensor\ket{k}$ and $\set{\ket{k}}_{k=0}^{N-1}$ is an arbitrary orthonormal basis for $\complex^N$. If and only if $\Phi$ is CP, then $J(\Phi)$ is positive semidefinite~\cite{C75}, and if and only if $\Phi$ is TP, then $\trace_A(J(\Phi))=I/N$~\cite{BZ06}. Here, $A$ and $B$ denote the respective subsystems of $J(\Phi)$.

Let us now state the link between the Jamio{\l}kowski representation for quantum channels and $\Q$, which will be the starting point for showing our result.

\begin{theorem}[Horodecki, Shor, and Ruskai~\cite{HSR03}]\label{def:EBmaps} Given a linear map $\Phi:L(\complex^{N})\mapsto L(\complex^M)$ that is CP and TP, i.e.\ $J(\Phi)$ is positive semidefinite and $\trace_A(J(\Phi))=I/N$, respectively, the following are equivalent:
\begin{enumerate}
    \item $\Phi$ is entanglement-breaking.
    \item $J(\Phi)$ is separable, i.e.\ $J(\Phi)\in\SMN$.\label{item:EB2}
\end{enumerate}
\end{theorem}

\noindent This theorem links $\Q$ and $\EBP$ in the following (immediate) way: The problem of determining whether a state $\rho\in\dens$ with $\trace_A (\rho)=I/N$ is separable is equivalent to the problem of determining whether the corresponding channel $\Phi$, such that $J(\Phi)=\rho$, is EB. Thus, if the former problem is NP-hard, so is the latter. To show that $\EBP$ is NP-hard, we are hence reduced to answering the question: Is $\Q$ still NP-hard if one is promised that the input state $\rho$ satisfies $\trace_A (\rho)=I/N$? Let us remark that, in contrast, the similar problem of determining whether a CP, but \emph{not} necessarily TP, map $\Phi$ is EB is trivially strongly NP-hard by Theorems~\ref{thm:finalthm} and~\ref{def:EBmaps}, since by the properties of the Jamio{\l}kowski representation discussed, dropping the TP constraint on $\Phi$ in Theorem~\ref{def:EBmaps} corresponds to dropping the constraint that $\trace_A (\rho)=I/N$.

We now sketch a proof that $\Q$ remains NP-hard even if one is promised the input state $\rho$ satisfies $\trace_A (\rho)=I/N$. Our approach is to demonstrate a poly-time many-one reduction to this problem from general $\Q$. Specifically, we show how to efficiently map an arbitrary input $\rho\in\dens$ for $\Q$ to some $\rho^\prime\in\denstwo$, such that $\trace_{A^\prime}(\rho^\prime)=I/N$ (where $A^\prime$ denotes subsystem $\complex^2\otimes \complex^M$), and $\rho$ is separable if and only if $\rho^\prime$ is separable across the $A^\prime/B$ split. The mapping proceeds as follows. Assume we have an efficient algorithm $Q$ for determining separability in the restricted case of $\trace_A(\rho)=I/N$. Then (each step below will be explained in the ensuing discussion):

\begin{enumerate}
%    \item If $\trace_A(\rho) = I/N$, input $\rho$ into $Q$ and return $Q$'s answer.\label{item:step1}
    \item Define the TPCP map $\Phi:\complex^M\otimes\complex^N\mapsto\complex^2\otimes\complex^M\otimes\complex^N$ such that
        \begin{equation*}
            \Phi(\sigma) := (1-p)\ketbra{0}{0}\otimes\sigma + p\ketbra{1}{1}\otimes I/N,
        \end{equation*}
        for $p=1-1/N$. \label{item:step2}
    \item Define the CP, but not necessarily TP, map $\Upsilon:\complex^2\otimes\complex^M\otimes\complex^N\mapsto\complex^2\otimes\complex^M\otimes\complex^N$ (omitting normalization for simplicity)
        \begin{equation*}
            \Upsilon(\sigma)=\left(I\otimes{\sigma_B^{-1/2}}\right)\sigma \left(I\otimes{\sigma_B}^{-1/2}\right),
        \end{equation*}
        where $\sigma_B=\trace_{A^\prime}(\sigma)$, $A^\prime$ and $B$ denote the subsystems $\complex^2\otimes\complex^M$ and $\complex^N$, respectively, and the identity $I$ acts on $A^\prime$.
    \item Call $Q$ with input $\Upsilon(\Phi(\rho))$, and return $Q$'s answer.\label{item:step3}

\end{enumerate}

To explain this reduction, we begin with Step 2. Our goal is to transform $\rho$ into a state $\rho^\prime$ with a maximally mixed subsystem, while preserving separability. One way to directly achieve this is to apply the local operations and classical communication (LOCC) map $\Upsilon$ to $\rho$, which locally applies the inverse of $\trace_A(\rho)$, ensuring $\trace_A(\Upsilon(\rho))=I/N$ (in this case $\Upsilon$ would act analogously on $\complex^M\otimes\complex^N$). As an aside, observe that while $\Upsilon$ is CP, it is not necessarily TP. Hence, $\Upsilon$ is a \emph{probabilistic} map (i.e.\ succeeds with some non-zero probability, and for this reason may be classified as a \emph{stochastic} LOCC map, or SLOCC), and in particular is an example of a \emph{local filter}, the latter being first considered in~\cite{G96}. To see that $\Upsilon$ preserves separability, note that $\Upsilon$ can be inverted with non-zero probability, and so $\Upsilon(\rho)$ is entangled if and only if $\rho$ is, since LOCC operations cannot create entanglement from scratch.

However, we cannot always begin by applying Step 2, because in general $\trace_A(\rho)$ is not full rank, and even if it is, it must have a low condition number in order to allow us to compute its inverse reliably when applying $\Upsilon$. Here, we define the \emph{condition number}~\cite{Hog07} of a matrix $C$ as $\kappa(C)=\norm{C}\norm{C^{-1}}$, where $\norm{\cdot}$ denotes the operator norm. To address this, we first preprocess $\rho$ in Step 1 by applying $\Phi$, which consists of two main components: mixing with the identity, and attaching a ``marker'' ancilla. The first component, mixing with the identity, ensures that not only is $\rho_B:=\trace_{A^\prime}(\Phi(\rho))$ full rank, but also that $\kappa(\rho_B)\leq 3$ for all $N\geq2$. To see this, observe that since $\rho_B$ is positive semidefinite, one has $\kappa(\rho_B) = \lambda_{\operatorname{max}}(\rho_B)/\lambda_{\operatorname{min}}(\rho_B)$, where $\lambda_{\operatorname{max}}(\rho_B)$ and $\lambda_{\operatorname{min}}(\rho_B)$ denote the largest and smallest eigenvalues of $\rho_B$ respectively. Combining this with the fact that $\rho_B = (1-p)\trace_A\rho + pI/N$ yields
\begin{equation}\label{eqn:cond}
    \kappa(\rho_B) = \frac{(1-p)\lambda_{\operatorname{max}}(\trace_A\rho) + p/N}{(1-p)\lambda_{\operatorname{min}}(\trace_A\rho) + p/N}\leq \frac{(1-p) + p/N}{p/N},
\end{equation}
where we have used the trivial bounds $\lambda_{\operatorname{max}}(\trace_A\rho) \leq 1$ and $\lambda_{\operatorname{min}}(\trace_A\rho)\geq 0$. For $p=1-1/N$, the right side of Eq.~(\ref{eqn:cond}) equals $(2N-1)/(N-1)$, which is upper bounded by $3$ for $N\geq 2$, as desired. The second component of $\Phi$, the ``marker'' ancilla, serves to preserve separability in a fashion similar to $\Upsilon$ by ensuring the LOCC operation $\Phi$ can be inverted with non-zero probability (which is achieved by performing a von Neumann measurement on the ancilla in the computational basis). We thus have that $\Phi(\rho)$ is entangled across the $A^\prime/B$ split if and only if $\rho$ is entangled across the $A/B$ split.

Having transformed $\rho$ into the state $\rho^\prime:=\Upsilon(\Phi(\rho))$ which satisfies our requirements for $Q$, we finally proceed to Step 3, in which we plug $\rho^\prime$ into $Q$ and return $Q$'s answer. Thus, the existence of an efficient algorithm $Q$ implies an efficient algorithm for $\Q$, yielding that determining separability of quantum states with a mixed subsystem is NP-hard. By Theorem~\ref{def:EBmaps}, this in turn implies our desired result, that $\EBP$ is NP-hard.

A few remarks are in order. First, to make this argument rigorous, we must use Weak Membership formulations, for which we will need to introduce additional notation. Let $\D:=\set{\rho\in\dens \mid \trace_A(\rho) = I/N}$, or equivalently, the convex set of operators which are the Jamio{\l}kowski representations of channels $\Phi:L(\complex^{N})\mapsto L(\complex^M)$, and let $\EB:=\D\cap \SMN$, or equivalently, the convex subset of EB channels in $\D$ (convexity of $\EB$ is clear from Def.~\ref{def:EBmapsreal}). In a rigorous argument, instead of reducing from $\Q$ to $\EBP$ as done above, one must more accurately reduce $\WMEM$ to $\WMEMEB$ (for the same reason regarding finite precision as in Section~\ref{scn:reduction}). As per Definition~\ref{def:WMEM}, this means we require that relative to $\D$, $\EB$ is be compact, well-bounded and p-centered. To see that $\EB$ is well-bounded and p-centered, observe that it follows directly from the fact that $\SMN$ is contained in a ball of radius $R=\sqrt{(MN-1)/MN}$ and contains a ball centered at $I/MN$ of radius $r=\sqrt{1/MN(MN-1)}$ (see Section~\ref{scn:def}) that $\EB$ must also be contained in a ball of the same radius $R$ and contain a ball\footnote{Note that such a contained ball about $I/MN$ does \emph{not} exist when $\EB$ is embedded into $\dens$ (as opposed to $\D$), since for any $\rho\in\dens$ with $\rho\not\in \D$, and for any $\epsilon > 0$, we have that $(1-\epsilon)I/MN + \epsilon\rho$ is not in $\D$.}~~of the same radius $r$. To show compactness relative to $\D$, we first have by the continuity and linearity of the partial trace that the complement of $\D$ is open, implying $\D$ is closed. Since $\SMN$ is compact and $\EB = \D\cap\SMN$, it follows that $\EB$ is compact relative to $\dens$, since in a metric space the intersection of a closed set and a compact set is compact (Corollary of Theorem 2.35 in~\cite{Rud76}). Finally, since $\EB\subseteq \D\subseteq \dens$, we have that $\EB$ is compact relative to $\D$ (Theorem 2.33 in~\cite{Rud76}). Hence, $\EB$ satisfies our requirements for Weak Membership.

To address the next technical detail, recall that Weak Membership takes as input real vectors, not operators. As discussed above, to ensure that $\EB$ is compact, well-bounded and p-centered, we wish to phrase $\EBP$ relative to $\D$, as opposed to $\dens$. To do so, we must map elements of $\D$ to real vectors in $\reals^{(M^2-1)N^2}$, rather than $\reals^{M^2N^2-1}$, where $(M^2-1)N^2$ and $M^2N^2-1$ are the dimensions of $\D$ and $\dens$, respectively (to see that the dimension of $\D$ is $(M^2-1)N^2$, observe that the partial trace constraint $\trace_A(\rho)=I/N$ effectively places $N^2-1$ linear constraints on $\rho\in\dens$). To achieve this mapping, we use what is sometimes called the Fano form~\cite{F83} for any $\rho\in\dens$:
\begin{equation}
    \rho = \frac{1}{MN}\left[I^A\otimes I^B + \frac{M}{2}\ve{r}^A\cdot\ve{\sigma}^A\otimes{I^B}+
                \frac{N}{2}I^A\otimes\ve{r}^B\cdot\ve{\sigma}^B+\frac{MN}{4}\sum_{i=0}^{M^2-1}\sum_{j=0}^{N^2-1}T_{ij}\sigma^A_i\otimes\sigma^B_j\right].\label{eqn:fano}
\end{equation}
Here, $\ve{r}^A$ and $\ve{\sigma}^A$ denote the real $(M^2-1)$-dimensional Bloch vector for subsystem $A$ and vector of generators for $SU(M)$, respectively, as defined in Eq.~(\ref{eqn:densityToBloch2}), and each real $T_{ij}$ can be defined $T_{ij}:=\trace(\sigma^A_i\otimes\sigma^B_j\rho)$ (this matrix $T$ is sometimes called the \emph{correlation} matrix). The definitions for subsystem $B$ are analogous. Since all $\rho\in\D$ have $\ve{r}^B=\ve{0}$, we can encode any $\rho\in\D$ directly using only the remaining $(M^2-1)N^2$ real values given by $\ve{r}^A$ and $T$, as desired.

Finally, we remark that in order to extend the result here to \emph{strong} NP-hardness of $\EBP$ using the improved hardness results of Theorem~\ref{thm:strong}, one must lower bound the distance $\beta^\prime$ of $\Upsilon(\Phi(\rho))$ from the border of $\EBtwo$, given the promise that $\rho$ is $\beta$ away from the border of $\SMN$ (to be clear, by \emph{distance} here we mean the Frobenius norm of the difference of the Jamio{\l}kowski representations of the maps). In particular, for inverse polynomial $\beta$, we desire inverse polynomial $\beta^\prime$ in order to prove an inverse polynomial hardness gap for $\EBP$. This may prove difficult, as $\Upsilon(\Phi(\rho))$ lives in a different space than $\rho$. We leave this as an open problem.

%======================================================================
\section{Concluding Comments}~\label{scn:conclusion}
%======================================================================
We have shown that the problem of Weak Membership over the set of separable quantum states $\SMN$ is strongly NP-hard, implying it is NP-hard even if one allows ``moderate'' error, i.e.\ $\beta\leq1/\poly(M,N)$. This shows that it is NP-hard to determine whether an arbitrary quantum state within an inverse polynomial distance from the separable set is entangled. It would be interesting to know whether $\WMEM$ remains NP-hard for $\beta$ a constant. Solving this problem would require a different approach than taken here, as, for example, even the first link of the reduction, $\C\leq_m \R$, introduces an inverse dependence on the dimension. We also remark that the hardness result shown here is via a \emph{Turing} reduction. It remains open, to the best of our knowledge, whether $\WMEM$ can be shown NP-hard under a \emph{many-one} reduction.

We next discussed two applications of our result: (1) We observed immediate lower bounds on the maximum Euclidean distance between a bound entangled state and $\SMN$, and (2) we showed that the problem of determining whether a quantum channel is entanglement-breaking is NP-hard. Whether this latter problem is \emph{strongly} NP-hard is left an an open question, for which we believe Theorem~\ref{thm:finalthm} should prove useful.

\section{Acknowledgements}
We thank the following: Ben Toner for ideas which spurred this line of thought, Richard Cleve and Stephen Vavasis for helpful discussions, John Watrous and Marco Piani for bringing local filters to our attention, and anonymous reviewers for their insightful comments in shaping this work. This work was partially supported by Canada's NSERC, CIAR and MITACS.

\emph{Note:} Shortly after we first posted the results herein, Beigi~\cite{B08} showed a reduction from 2-out-of-4 SAT which can also be used to obtain a result similar to Theorem~\ref{thm:strong}. The question of NP-hardness for constant $\beta$, however, unfortunately remains open.

\appendix
\section{Appendix}\label{app:A}
\begin{lemma}\label{l:appendix1}
    Given density operators $\rho_1, \rho_2\in\dens$, with
    corresponding Bloch vectors $\vec{\alpha}$ and $\vec{\beta}$ given by equations $\rho_1=\frac{I}{MN}+\frac{1}{2}\vec{\alpha}\cdot\vec{\sigma}$ and $\rho_2=\frac{I}{MN}+\frac{1}{2}\vec{\beta}\cdot\vec{\sigma}$ (as per Eq.~(\ref{eqn:densityToBloch})),
    respectively, we have
    $\enorm{\rho_1-\rho_2}=\frac{1}{\sqrt{2}}\enorm{\vec{\alpha}-\vec{\beta}}$.
\end{lemma}
\begin{proof}
    Via straightforward manipulation and the fact that $\trace(\sigma_i\sigma_j)=2\delta_{ij}$, we have:
    \begin{eqnarray*}
        \enorm{\rho_1-\rho_2}&=&%\enorm{\left(\frac{I}{N}+\frac{1}{2}\sum_{i=1}^{N^2-1}\alpha_i\sigma_i\right)-\left(\frac{I}{N}+\frac{1}{2}\sum_{j=1}^{N^2-1}\beta_j\sigma_j\right)}\\
                           %&=&
                           \frac{1}{2}\enorm{\sum_{i=1}^{M^2N^2-1}(\alpha_i-\beta_i)\sigma_i}\\
                           &=&\frac{1}{2}\sqrt{\trace\left[\left(\sum_{i=1}^{M^2N^2-1}(\alpha_i-\beta_i)\sigma_i\right)^\adjoint\left(\sum_{j=1}^{M^2N^2-1}(\alpha_j-\beta_j)\sigma_j\right)\right]}\\
                           &=&\frac{1}{2}\sqrt{\sum_{i=1,j=1}^{M^2N^2-1}(\alpha_i-\beta_i)(\alpha_j-\beta_j)2\delta_{ij}}\\
                           &=&\frac{1}{\sqrt{2}}||\vec{\alpha}-\vec{\beta}||_2.
    \end{eqnarray*}
\end{proof}

\begin{lemma}\label{l:appendix2}
    Combining Theorem~\ref{thm:link1} and Lemma~\ref{l:link2} gives $\ec\in O(\sqrt{N})$.
\end{lemma}
\begin{proof}
    By definition, we have:
    \begin{equation}\label{eqn:cnorm}
        \ec=\sqrt{\sum_{i=1}^{m}\left[\frac{1}{2}\trace(C\sigma_i)\right]^2},
    \end{equation}
    where $m=M^2N^2-1$. Recall now the definition of $C$ from Eq.~(\ref{eqn:Cmatrix}), where in our case, each $A_i\in\reals^{N\times N}$ is all zeroes except for its upper left corner, which is set to submatrix $B_i\in\reals^{l\times l}$ from Theorem~\ref{thm:link1}. Each $B_i$ in turn is all zeroes, except for some index $(k,l)$ (and hence $(l,k)$, by symmetry), $1\leq k < l \leq n$, which is set to the $(k,l)$th entry of the adjacency matrix $A_G$ of graph $G$ (see Theorem~\ref{thm:link1} and ensuing discussion). We also require an explicit construction for the generators $\sigma_i$ of $SU(MN)$, given for example in~\cite{Kim03}, where $\set{\sigma_i}_{i=1}^{M^2N^2-1}=\set{U_{pq},V_{pq},W_{r}}$, such that for $1\leq p<q\leq MN$ and $1\leq r \leq MN-1$, and $\set{\ve{x}_i}_{i=1}^{MN}$ an orthonormal basis for Hilbert space $\mathcal{H}^{MN}$:
    \begin{eqnarray}
        U_{pq}&=&\ve{x}_p\ve{x}_q^\adjoint+\ve{x}_q\ve{x}_p^\adjoint\\
        V_{pq}&=&-i\ve{x}_p\ve{x}_q^\adjoint+i\ve{x}_q\ve{x}_p^\adjoint\\
        W_{r} &=&\sqrt{\frac{2}{r(r+1)}}\left[\left(\sum_{k=1}^{r}\ve{x}_k\ve{x}_k^\adjoint\right)-r\ve{x}_{r+1}\ve{x}_{r+1}^\adjoint\right].
    \end{eqnarray}

    Due to the symmetry of $C$ and the fact that $\trace(C)=0$, it is clear that only the generators of the form $U_{pq}$ will contribute to the sum in Eq.~(\ref{eqn:cnorm}). Further, for each edge in $G$, $\trace(CU_{pq})=2$ for each $U_{pq}$ whose non-zero indices match those of the entries in $C$ corresponding to that edge. Since each edge contributes four (symmetrically placed) entries of $1$ to $C$, we hence have $\ec=\frac{1}{2}\sqrt{(2\hat{e})2^2}=\sqrt{2\hat{e}}$, where $\hat{e}$ denotes the number of edges in $G$. Since $\hat{e}\in O(n^2)$ ($n$ the number of vertices in $G$), and $N\in \Theta(n^2)$, we have $\ec\in O(\sqrt{N})$, as required.
\end{proof}
\bibliographystyle{unsrt}
\bibliography{BibSep}

\begin{thebibliography}{10}

\bibitem{S35}
E.~Schr\"{o}dinger.
\newblock Discussion of probability relations between separated systems.
\newblock {\em Proc. Cambridge Phil. Soc.}, 31:555--563, 1935.

\bibitem{BBCJPW93}
C.~Bennett, G.~Brassard, C.~Cr\'{e}peau, R.~Josza, A.~Peres, and W.~K.
  Wootters.
\newblock Teleporting an unknown quantum state via dual classical and
  {E}instein-{P}odolsky-{R}osen channels.
\newblock {\em Phys. Rev. Lett.}, 70:1895---1899, 1993.

\bibitem{BW92}
C.~H. Bennett and S.~J. Wiesner.
\newblock Communication via one- and two-particle operators on
  {E}instein-{P}odolsky-{R}osen states.
\newblock {\em Phys. Rev. Lett.}, 69:2881--2884, 1992.

\bibitem{S94}
P.~Shor.
\newblock Algorithms for quantum computation: discrete logarithms and
  factoring.
\newblock In {\em Proc. of 35th Ann. Symp. on Found. of Comp. Sci.}, pages
  124--134, 1994.

\bibitem{CB97}
R.~Cleve and H.~Buhrman.
\newblock Substituting quantum entanglement for communication.
\newblock {\em Phys. Rev. A}, 56:1201--1204, 1997.

\bibitem{G97}
L.~K. Grover.
\newblock Quantum telecomputation.
\newblock arXiv:quant-ph/9704012v2, 1997.

\bibitem{E91}
A.~Ekert.
\newblock Quantum cryptography based on {B}ell's theorem.
\newblock {\em Phys. Rev. Lett.}, 67:661--663, 1991.

\bibitem{P96}
A.~Peres.
\newblock Separability criterion for density matrices.
\newblock {\em Phys. Rev. Lett.}, 77(8), 1996.

\bibitem{HHH96}
M.~Horodecki, P.~Horodecki, and R.~Horodecki.
\newblock Separability of mixed states: necessary and sufficient conditions.
\newblock {\em Phys. Lett. A}, 223(1):1--8, 1996.

\bibitem{DHR02}
M.~J. Donald, M.~Horodecki, and O.~Rudolph.
\newblock The uniqueness theorem for entanglement measures.
\newblock {\em J. Math. Phys.}, 43:4252--4272, 2002.

\bibitem{BDSW96}
C.~H. Bennett, D.~P. DiVincezo, J.~Smolin, and W.~K. Wootters.
\newblock Mixed-state entanglement and quantum error correction.
\newblock {\em Phys. Rev. A}, 54:3824, 1996.

\bibitem{VPRK97}
V.~Vedral, M.~B. Plenio, M.~A. Rippin, and P.~L. Knight.
\newblock Quantifying entanglement.
\newblock {\em Phys.~Rev.~Lett.}, 78:2275, 1997.

\bibitem{VW02}
G.~Vidal and R.~F. Werner.
\newblock Computable measure of entanglement.
\newblock {\em Phys. Rev. A}, 65:032314, 2002.

\bibitem{Bru02}
D.~Bru\ss.
\newblock Characterizing entanglement.
\newblock {\em J. of Math. Phys.}, 43(9), 2002.

\bibitem{HHH07}
R.~Horodecki, P.~Horodecki, M.~Horodecki, and K.~Horodecki.
\newblock Quantum entanglement.
\newblock arXiv:quant-ph/0702225v2, 2007.

\bibitem{BZ06}
I.~Bengtsson and K.~\.{Z}yczkowski.
\newblock {\em Geometry of Quantum States: An Introduction to Quantum
  Entanglement}.
\newblock Cambridge University Press, 2006.

\bibitem{Gur03}
L.~Gurvits.
\newblock Classical deterministic complexity of edmonds' problem and quantum
  entanglement.
\newblock In {\em Proc. of the 35th ACM symp. on Theory of comp.}, pages
  10--19, New York, 2003. ACM Press.

\bibitem{Kim03}
G.~Kimura.
\newblock The {B}loch vector for {N}-level systems.
\newblock {\em Phys. Lett. A}, 314(5):339--349, 2003.

\bibitem{GJ79}
M.~R. Garey and D.~S.Johnson.
\newblock {\em Computers and Intractability: A Guide to the theory of
  NP-completeness}.
\newblock W.~H.~Freeman and Company, New York, 1979.

\bibitem{Iou07}
L.~Ioannou.
\newblock Computational complexity of the quantum separability problem.
\newblock {\em Quant. Inf. and Comp.}, 7(4):335, 2007.

\bibitem{GO09}
L.~Gurvits and A.~Olshevsky.
\newblock On the {NP}-hardness of checking matrix polytope stability and
  continuous-time switching stability.
\newblock {\em IEEE Trans. Auto. Control}, 54:337--341, 2009.

\bibitem{Gro88}
M.~Gr\"{o}tschel, L.~Lov\'{a}sz, and A.~Schrijver.
\newblock {\em Geometric Algorithms and Combinatorial Optimization}.
\newblock Springer-Verlag, Berlin, 1988.

\bibitem{YN76}
D.~B. Yudin and A.~S. Nemirovskii.
\newblock Informational complexity and efficient methods for the solution of
  convex extremal problems (in {R}ussian).
\newblock {\em Ekonomica i Matematicheskie Metody}, 12:357–--369, 1976.

\bibitem{Liu07}
Y.~K. Liu.
\newblock The complexity of the consistency and {N}-representability problems
  for quantum states.
\newblock arXiv.org:0712.3041v1, 2007.
\newblock PhD Thesis.

\bibitem{Ber04}
D.~Bertsimas and S.~Vempala.
\newblock Solving convex programs by random walks.
\newblock {\em J. of the ACM}, 51(4), 2004.

\bibitem{HHH98}
M.~Horodecki, P.~Horodecki, and R.~Horodecki.
\newblock Mixed-state entanglement and distillation: Is there a ``bound''
  entanglement in nature?
\newblock {\em Phys. Rev. Lett.}, 80(24):5239--5242, 1998.

\bibitem{J72}
A.~Jamio{\l}kowski.
\newblock Linear transformations which preserve trace and positive
  semi-definiteness of operators.
\newblock {\em Rep.~Math.~Phys.}, 3:275, 1972.

\bibitem{HSR03}
M.~Horodecki, P.~W. Shor, and M.~B. Ruskai.
\newblock General entanglement breaking channels.
\newblock {\em Rev. Math. Phys.}, 15:629, 2003.

\bibitem{Gur02}
L.~Gurvits and H.~Barnum.
\newblock Largest separable balls around the maximally mixed bipartite quantum
  state.
\newblock {\em Phys. Rev. A}, 66(6):062311, 2002.

\bibitem{V89}
E.~B. Vinberg.
\newblock {\em Linear Representations of Groups}.
\newblock Birkh\"{a}user Basel, 1989.

\bibitem{MS65}
T.~S. Motzkin and E.~G. Straus.
\newblock Maxima for graphs and a new proof of a theorem of {T}\'{u}ran.
\newblock {\em Canadian J. Math.}, 17:533–--540, 1965.

\bibitem{Ruskai03}
M.~B. Ruskai.
\newblock Qubit entanglement breaking channels.
\newblock {\em Rev. Math. Phys.}, 15:643, 2003.

\bibitem{King02}
C.~King.
\newblock Maximization of capacity and lp norms for some product channels.
\newblock {\em J. Math. Phys.}, 43:1247, 2002.

\bibitem{Shor02}
P.~Shor.
\newblock Additivity of the classical capacity of entanglement-breaking quantum
  channels.
\newblock {\em J. Math. Phys.}, 43:4334, 2002.

\bibitem{King03}
C.~King.
\newblock Maximal p-norms of entanglement breaking channels.
\newblock {\em Quant. Inf. and Comp.}, 3:186, 2003.

\bibitem{HolevoSW05}
A.~S. Holevo, M.~E. Shirokov, and R.~F. Werner.
\newblock Separability and entanglement-breaking in infinite dimensions.
\newblock {\em Russian. Math Surveys}, 60, 2005.

\bibitem{Holevo08}
A.~S. Holevo.
\newblock Entanglement-breaking channels in infinite dimensions.
\newblock {\em Problems Inf. Transmiss.}, 44:171--184, 2008.

\bibitem{C75}
M.-D. Choi.
\newblock Completely positive linear maps on complex matrices.
\newblock {\em Linear Alg.~Appl.}, 10:285, 1975.

\bibitem{G96}
N.~Gisin.
\newblock Hidden quantum nonlocality revealed by local filters.
\newblock {\em Phys. Lett. A}, 210(3):151--156(6), 1996.

\bibitem{Hog07}
L.~Hogben (editor).
\newblock {\em Handbook of Linear Algebra}.
\newblock Chapman \& Hall/CRC, 2007.

\bibitem{Rud76}
W.~Rudin.
\newblock {\em Principles of Mathematical Analysis}.
\newblock McGraw-Hill, 1976.

\bibitem{F83}
U.~Fano.
\newblock Pairs of two-level systems.
\newblock {\em Rev. Mod. Phys.}, 55:855--874, 1983.

\bibitem{B08}
S.~Beigi.
\newblock {NP} vs {QMA\_log(2)}.
\newblock arXiv:0810.5109v1, 2008.

\end{thebibliography}

\end{document}